\documentclass{amsart}

\usepackage{amssymb,amsmath}
\usepackage{verbatim}

\usepackage[section]{algorithm}
\usepackage{algpseudocode}

\newtheorem{theorem}{Theorem}[section]
\newtheorem{proposition}[theorem]{Proposition}

\newtheorem{corollary}[theorem]{Corollary}

\newtheorem{definition}[theorem]{Definition}

\newtheorem{remark}[theorem]{Remark}

\newtheorem{example}[theorem]{Example}

\def\NN{{\mathbb N}}
\def\ZZ{{\mathbb Z}}
\def\QQ{{\mathbb Q}}
\def\KK{{\mathbb K}}
\def\SS{{\mathbb S}}
\def\FF{{\mathbb F}}

\def\lcm{{\mathrm{lcm}}}
\def\Mon{{\mathrm{Mon}}}
\def\NF{{\mathrm{NF}}}

\def\Ker{{\mathrm{Ker\,}}}

\def\Gr{Gr\"obner}

\def\lm{{\mathrm{lm}}}
\def\lc{{\mathrm{lc}}}
\def\lt{{\mathrm{lt}}}
\def\LM{{\mathrm{LM}}}

\def\cC{{\mathcal{C}}}

\def\hvarphi{{\hat{\varphi}}}
\def\id{{\mathrm{id}}}
\def\GF{{\mathrm{GF}}}
\def\M{{\mathrm{M}}}

\def\Tri{{\sc Trivium}}
\def\Biv{{\sc Bivium}}
\def\Kee{{\sc KeeLoq}}

\def\bX{{\bar{X}}}
\def\bR{{\bar{R}}}

\def\bT{{\bar{\mathrm{T}}}}
\def\bS{{\bar{\mathrm{S}}}}

\def\T{{\mathrm{T}}}
\def\S{{\mathrm{S}}}

\def\bV{{\bar{V}}}

\def\cO{{\mathcal{O}}}

\begin{document}

\title[Stream/block ciphers, difference equations and algebraic attacks] 
{Stream/block ciphers, difference equations and algebraic attacks}

\author[R. La Scala]{Roberto La Scala$^*$}

\author[S.K. Tiwari]{Sharwan K. Tiwari$^{**}$}

\address{$^*$ Dipartimento di Matematica, Universit\`a di Bari,
Via Orabona 4, 70125 Bari, Italy}
\email{roberto.lascala@uniba.it}

\address{$^{**}$ Scientific Analysis Group, Defence Research
\& Development Organization, Metcalfe House, Delhi-110054, India}
\email{shrawant@gmail.com}

\thanks{The first author acknowledges the support of the University of Bari,
Grant ref. 73251. The second author thanks the Scientific Analysis Group, DRDO,
Delhi, for the support.}



\begin{abstract}
In this paper we model a class of stream and block ciphers as systems of (ordinary)
explicit difference equations over a finite field. We call this class ``difference ciphers''
and we show that ciphers of application interest, as for example systems of LFSRs
with a combiner, \Tri\ and \Kee, belong to the class. By using Difference Algebra,
that is, the formal theory of difference equations, we can properly define and study
important properties of these ciphers, such as their invertibility and periodicity.
We describe then general cryptanalytic methods for difference ciphers that follow
from these properties and are useful to assess the security. We illustrate such algebraic
attacks in practice by means of the ciphers \Biv\ and \Kee.
\end{abstract}

\maketitle

%%%%%%%%%%%%%%%%%%%%%%%%%%%%%%%%%%%%%%%%%%%%%%%%%%%%%%%%%%%%%%%%%%%%%%

\section{Introduction}

Algebraic cryptanalysis concerns the possibility to perform an attack to a cipher
as an instance of polynomial system solving. This idea dates back at least to Shannon
who in his foundational paper of modern Cryptography \cite{Sh} wrote that security
can be essentially assumed ``if we could show that solving a certain (crypto)system
requires at least as much work as solving a system of simultaneous equations in a large
number of unknowns, of a complex type''. In the last 20 years, these cryptanalytic ideas
have become actual algorithms and implementations by the work of many cryptographers:
see, for instance, the long reference list of the comprehensive book of G. Bard
entitled ``Algebraic Cryptanalysis'' \cite{Ba}. Among this vast literature, we would like
to mention the pioneering work of N. Courtois \cite{CM,CP} and his methods for
polynomial system solving, such as XL and ElimLin \cite{CKPS}. It is important to reference
also the role played by the F4-F5 algorithms due to J.-C. Faug\`ere \cite{Fa1,Fa2} for providing
complexity formulas for \Gr\ bases computations over (semi)regular systems \cite{BFP,BFS}.
Indeed, such complexity is much lower than in the worst case and these systems
naturally arise in the context of multivariate cryptography. In other words,
such complexity formulas provide essential cryptanalysis for multivariate protocols
as Rainbow \cite{CDKPPSY} which is among the three NIST post-quantum signature finalists.

The usual formalism of algebraic cryptanalysis are systems of polynomial equations
and generally Commutative Algebra (over finite fields) but, as a matter of fact,
a natural modeling of many stream and block ciphers are systems of (algebraic
ordinary) explicit difference equations. In fact, such ciphers are defined
as recursive rules determining the evolution of a vector, with entries in a finite field,
which is called the state or register of the cipher. This evolution runs along
a discrete time corresponding to clocks or rounds. Because of their simple structure,
explicit difference systems provide such a recursion, that is, the existence and uniqueness
of solutions once given their initial state. Despite this simplicity, the solutions
of these systems can be extremely involved and hence interesting for the purposes
of Cryptography. We mention, for instance, the complex evolution of a discrete
predator-prey model \cite{Ag} or the security properties of the stream cipher
\Tri\ \cite{DCP} where both underlying systems are just few quadratic explicit
difference equations.

The formal theory of algebraic difference equations is called Difference Algebra
and it was introduced by J.F. Ritt \cite{RR} as a discrete couterpart of his
celebrated Differential Algebra. This theory provides important insights about
the structure of the solutions of a system of difference equations (see \cite{Co,Le,Wi,Wi2}).
In particular, by mimicking Commutative Algebra, one has the notion of difference
ideal and corresponding difference variety, as well the methods of difference
\Gr\ bases \cite{GLS,LS}. The use of difference equations in Cryptography has a long
tradition because of Linear Feedback Shift Registers (briefly LFSRs) which are just
linear difference equations over the field with two elements. These very simple recursive
rules have been used at least since the 1950s to obtain, for instance, pseudorandom
number generators. For a classical reference, see \cite{Go}.
The problem with LFSRs is that they are easy to attack because of linearity
and hence a standard strategy to enhance security consists in considering
a system of LFSRs together with a combining (or filtering) non-linear function.
Many pseudorandom generators and stream ciphers are of this kind such as
the Geffe generator and the cipher E0 of the Bluetooth protocol. We reference
the book \cite{Kl} for a detailed guide to them.

Another option consists in having non-linear explicit difference equations
governing the evolution of the internal state of the cipher and a linear function
for defining the keystream output. The prototype of such stream cipher is \Tri.
Note that its explicit system has a state transition function which is invertible.
This is a possible flaw because some opponents may recover the initial state
containing the key by attacking any internal state. Despite many attacks
also of this kind \cite{EVP,HL,MB}, the three quadratic explicit difference equations
of \Tri\ remain inexpugnable. On the other hand, if an invertible system contains
a subsystem that can be used to evolve independently the keys, this flaw becomes
a resource for defining a block cipher. These ideas appear in the definition
of the block cipher \Kee\ which has been cryptanalized in a critical way
\cite{CBB,CBW} because of the short period of the state transition function
of its key subsystem.

The present paper is organized as follows. In Section 2, for any base field,
we introduce the formalism of Difference Algebra for the purpose of providing
a precise definition of systems of (ordinary) explicit difference equations
from an algebraic viewpoint. We introduce then the notion of state transition
endomorphism and we apply the methods of difference \Gr\ bases of difference
ideals \cite{GLS,LS} to obtain a key property contained in Theorem \ref{main}.
This property, together with Theorem \ref{keyeq}, will be used in Section 5
to obtain the equations satisfied by the keys for a general algebraic attack
to a difference stream cipher.

In Sections 3 and 4, by means of the notion of state transition endomorphism 
we define key properties of explicit systems such as their invertibility
and periodicity. The use of symbolic computation, namely \Gr\ bases,
provides effective methods to check invertibility and compute inverse
systems which is essential for the cryptanalysis of difference ciphers.
In Section 4 we also review, from an algebraic point of view, the theory
to maximize the period of a system of LFSRs over a prime finite field.

In Section 5, we finally introduce the class of ``difference stream and block ciphers''
as ciphers that are defined by explicit difference systems over a finite field.
The motivation is that many ciphers of application interests can be easily modeled
in this way and we can provide general methods for the algebraic cryptanalysis 
of this class. This theory can be used therefore for developing new ciphers.
As already mentioned, in Section 5 we obtain the equations satisfied by the keys
of a difference stream cipher which are compatible with a known keystream.
We show that a finite number of elements of the keystream are enough
to provide these equations. By means of the invertibility property studied
in Section 3, we also suggest how to reduce the degree of the equations.
We discuss then guess-and-determine strategies for polynomial system solving
over a finite field and provide details for the case that the system has
a single solution. Indeed, this is often the case in the cryptanalysis of
stream ciphers once a sufficient amount of elements of the keystream is known.
For difference block ciphers we discuss a general algebraic attack
which follows from the condition that the key subsystem has a short period.
This attack has been introduced for the cipher \Kee\ in the papers \cite{CBB,CBW}.

Finally, in Section 6 and 7 we illustrate the above algebraic attacks by applying
them to concrete ciphers, namely \Biv\ and \Kee. Although this is not our primary goal
in the present paper, we are able to obtain speedup with respect to similar attacks
to \Biv\ \cite{EVP}. We also improve the previous polynomial system solving time
\cite{CBB,CBW} in the attack to \Kee. An interesting comparison of \Gr\ bases vs
SAT solvers is also provided in our tests. Some conclusions are draw in Section 8.

\section{Explicit difference system}

Let $\KK$ be any field and fix an integer $n > 0$. Consider a set of variables
$X(t) = \{x_1(t),\ldots,x_n(t)\}$, for any $t\in\NN = \{0,1,2,\ldots\}$.
Put $X = \bigcup_{t\geq 0} X(t)$ and denote by $R = \KK[X]$ the polynomial
algebra in the infinite set of variables $X$. Moreover, consider the injective
algebra endomorphism $\sigma: R\to R$ such that $x_i(t)\mapsto x_i(t+1)$
for all $1\leq i\leq n$ and $t\geq 0$. We call $\sigma$ the {\em shift map of $R$}.
The algebra $R$, endowed with the map $\sigma$, is called the {\em algebra
of ordinary difference polynomials (with constant coefficients)}. We also
need the following notations. For any integers $r_1,\ldots,r_n\geq 0$
and $t\geq 0$, we define the subset
\[
\bX = \{x_1(0),\ldots,x_1(r_1-1),\ldots,x_n(0),\ldots,x_n(r_n-1)\}\subset X
\]
and the subalgebra $\bR = \KK[\bX]\subset R$.

\begin{definition}
Let $r_1,\ldots,r_n\geq 0$ be integers and consider some polynomials
$f_1,\ldots,f_n\in \bR$. A {\em system of (algebraic ordinary)
explicit difference equations} is by definition an infinite system of polynomial
equations of the kind
\begin{equation*}
\label{sys}
\left\{
\begin{array}{ccc}
x_1(r_1 + t) & = & \sigma^t(f_1), \\
& \vdots \\
x_n(r_n + t) & = & \sigma^t(f_n). \\
\end{array}
\right.
\quad (t\geq 0)
\end{equation*}
Such a system is denoted briefly as
\begin{equation}
\label{dsys}
\left\{
\begin{array}{ccc}
x_1(r_1) & = & f_1, \\
& \vdots \\
x_n(r_n) & = & f_n. \\
\end{array}
\right.
\end{equation}
A {\em $\KK$-solution of the system $(\ref{dsys})$} is an $n$-tuple
of functions $(a_1,\ldots,a_n)$ where each $a_i:\NN\to\KK$ satisfies
the equations $x_i(r_i + t) = \sigma^t(f_i)$, for all $t\geq 0$.
The element $a_i(t)\in\KK$ is called the {\em value of the function $a_i$
at the clock $t\geq 0$}.
\end{definition}

\begin{definition}
Consider an explicit difference system $(\ref{dsys})$. We define the algebra
endomorphism $\bT:\bR\to\bR$ such that, for any $i = 1,2,\ldots,n$
\[
x_i(0)\mapsto x_i(1),\ldots,x_i(r_i-2)\mapsto x_i(r_i-1),
x_i(r_i-1)\mapsto f_i.
\]
If $r = r_1 + \ldots + r_n$, we denote by $\T:\KK^r\to \KK^r$ the polynomial
map corresponding to $\bT$. For any polynomial $f\in \bR$ and for each vector
$v\in\KK^r$, one has that
\begin{equation}
\label{fund}
\bT(f)(v) = f(\T(v)).
\end{equation}
If $(a_1,\ldots,a_n)$ is a $\KK$-solution of $(\ref{dsys})$, we call the vector
\[
v(t) = (a_1(t),\ldots,a_1(t+r_1-1),\ldots,a_n(t),\ldots,a_n(t+r_n-1))\in\KK^r
\]
the {\em state of $(a_1,\ldots,a_n)$ at the clock $t\geq 0$}. In particular,
$v(0)$ is the {\em initial state of $(a_1,\ldots,a_n)$}. Then, the function
$\T$ maps the $t$-state $v(t)$ into the $(t+1)$-state $v(t+1)$, for all clocks
$t\geq 0$. We call $\bT$ the {\em state transition endomorphism} and $\T$
the {\em state transition map of the explicit difference system (\ref{dsys})}.
\end{definition}

\begin{example}

Let $\KK = \QQ$ and consider the following explicit system
\[
\left\{
\begin{array}{ccl}
x(1) & = & x(0)^2 + y(0)^2, \\
y(1) & = & 2 x(0)y(0).
\end{array}
\right.
\]
If $\bR = \KK[x(0),y(0)]$, the corresponding state transition endomorphism
$\bT:\bR\to\bR$ is defined as $x(0)\mapsto x(0)^2 + y(0)^2, y(0)\mapsto 2x(0)y(0)$.
One computes that the $\KK$-solutions of the above system are the functions $a,b:\NN\to\KK$
such that, for all $t\geq 0$
\begin{equation*}
\begin{gathered}
a(t) = \frac{(a(0) + b(0))^{2^t} + (a(0) - b(0))^{2^t}}{2}, \\
b(t) = \frac{(a(0) + b(0))^{2^t} - (a(0) - b(0))^{2^t}}{2}.
\end{gathered}
\end{equation*}
\end{example}

We have the following existence and uniqueness theorem for the solutions
of an explicit system.

\begin{theorem}
\label{unique}
Denote by $V_\KK$ the set of all $\KK$-solutions of the explicit difference system 
$(\ref{dsys})$. We have a bijective map $\iota:V_\KK\to\KK^r$ such that
\[
(a_1,\ldots,a_n)\mapsto (a_1(0),\ldots,a_1(r_1-1),\ldots,a_n(0),\ldots,a_n(r_n-1)).
\]
In other words, the system $(\ref{dsys})$ has a unique $\KK$-solution once
fixed its initial state. Moreover, the maps $\iota,\iota^{-1}$ are both
polynomial ones.
\end{theorem}

\begin{proof}
Consider the state transition map $\T:\KK^r\to\KK^r$ of (\ref{dsys}) which is
a polynomial map. Observe that all powers $\T^t:\KK^r\to\KK^r$ ($t\geq 0$) are
also polynomial maps. If
\[
v(t) = (a_1(t),\ldots,a_1(t+r_1-1),\ldots,a_n(t),\ldots,a_n(t+r_n-1))
\]
denotes the $t$-state of a $\KK$-solution $(a_1,\ldots,a_n)\in V_\KK$, the inverse map
\[
\iota^{-1}:v(0)\mapsto (a_1,\ldots,a_n)
\]
is obtained in the following way. The value $a_1(t)$ is the first coordinate
of the vector $v(t) = \T^t(v(0))$, $a_2(t)$ is its $(r_1+1)$-th coordinate
and so on. Since projections and $\T^t$ are polynomial maps, we conclude that
$\iota^{-1}$ is also such a map.
\end{proof}

Consider the state transition endomorphism $\bT:\bR\to\bR$ of the system (\ref{dsys}).
Note that all powers $\bT^t:\bR\to\bR$ ($t\geq 0$) are also endomorphisms whose
corresponding polynomial maps are the functions $\T^t:\KK^r\to\KK^r$.
For all $1\leq i\leq n$ and $t\geq 0$, we define the polynomial
\[
f_{i,t} = \bT^t(x_i(0))\in \bR.
\]
By the argument of Theorem \ref{unique} and the identity (\ref{fund}),
it follows that $a_i(t) = f_{i,t}(v(0))$, for all $\KK$-solutions
$(a_1,\ldots,a_n)\in V_\KK$.

We briefly introduce now the notion of difference \Gr\ basis which provides
very often an alternative way to compute the polynomial $f_{i,t}$.
For a complete reference we refer to \cite{GLS,LS}.

\begin{definition}
Let $I$ be an ideal of the algebra $R$. We call $I$ a {\em difference ideal}
if $\sigma(I)\subset I$. Denote $\Sigma = \{\sigma^t\mid t\geq 0\}$ and
let $G$ be a subset of $R$. Then, we define $\Sigma(G) = \{\sigma^t(g)\mid 
g\in G, t\geq 0\}\subset R$. We call $G$ a {\em difference basis} of a
difference ideal $I$ if $\Sigma(G)$ is a basis of $I$ as an ideal of $R$.
In other words, all elements $f\in I$ are such that
$f = \sum_i f_i \sigma^{t_i}(g_i)$ where $f_i\in R, g_i\in G$ and
$t_i\geq 0$. In this case, we denote
$\langle G\rangle_\sigma = \langle \Sigma(G) \rangle = I$.
\end{definition}

Consider an explicit difference system (\ref{dsys}) and define the subset
\[
G = \{x_1(r_1) - f_1, \ldots, x_n(r_n) - f_n\}\subset R.
\]
If $I = \langle G \rangle_\sigma$, we have that $(a_1,\ldots,a_n)$
is a $\KK$-solution of the system (\ref{dsys}) if and only if this is
a simultaneous $\KK$-solution of all polynomials $f\in I$. In other words,
by substituting the variables $x_i(t)$ of each $f\in I$ with the elements
$a_i(t)\in\KK$, one always obtains zero. Then, we also say that
$(a_1,\ldots,a_n)$ is a {\em $\KK$-solution of the difference ideal $I$}
and we put $V_\KK(I) = V_\KK$. For defining \Gr\ bases, one needs to introduce
monomial orderings on $R$.

\begin{definition}
\label{monord}
Let $\prec$ be a total ordering on the set $M = \Mon(R)$ of all monomials 
of $R$. We call $\prec$ a {\em monomial ordering of $R$} if the following
properties hold:
\begin{itemize}
\item[(i)] $\prec$ is a multiplicatively compatible ordering, that is,
if $m'\prec m''$ then $m m' \prec m m''$, for all $m, m', m''\in M$;
\item[(ii)] $\prec$ is a well-ordering, that is, every non-empty subset
of $M$ has a minimal element.
\end{itemize}
In this case, it follows that
\begin{itemize}
\item[(iii)] $1\prec m$, for all $m\in M, m\neq 1$.
\end{itemize}
\end{definition}

Indeed, from $1\succ m$ and the property (i) it follows that we have an infinite strictly
decreasing sequence
\[
1\succ m\succ m^2\succ \ldots
\]
which contradicts the property (ii). Even though the variables set $X$ is infinite,
by Higman's Lemma \cite{Hi} the polynomial algebra $R = K[X]$ can always be
endowed with a monomial ordering. For the following version of this key result,
see for instance \cite{AH}, Corollary 2.3 and also the remarks at the beginning
of page 5175 of that reference.

\begin{proposition}
Let $\prec$ be a total ordering on $M$ which verifies the properties
$(i),(iii)$ of Definition \ref{monord}. If the restriction of $\prec$
to the variables set $X\subset M$ is a well-ordering then $\prec$ is also
a well-ordering on $M$, that is, it is a monomial ordering of $R$.
\end{proposition}

To introduce difference \Gr\ bases, we need monomial orderings that are compatible
with the shift map. 

\begin{definition}
Let $\prec$ be a monomial ordering of $R$. We call $\prec$ a {\em difference
monomial ordering of $R$} if $m\prec m'$ implies that $\sigma(m)\prec \sigma(m')$,
for all $m,m'\in M$.
\end{definition}

Note that if $\prec$ is a difference monomial ordering, we have that
$m\prec\sigma(m)$, for all $m\in M, m\neq 1$. Indeed, by assuming $m\succ\sigma(m)$
one obtains an infinite strictly decreasing sequence
\[
m\succ\sigma(m)\succ\sigma^2(m)\succ \ldots
\]
which contradicts the property of $\prec$ of being a well-ordering.

An important class of difference monomial orderings can be defined
in the following way. Recall that all polynomial algebras $R(t) = K[X(t)]$
($t\geq 0$) are in fact isomorphic by means of the shift map. Then,
let us consider the same monomial ordering for all such algebras.
Since $R = \bigotimes_{t\geq 0} R(t)$, we can define on $R$ the product
monomial ordering such that $X(0)\prec X(1)\prec \ldots$. For any choice
of a monomial ordering on $R(0)$, this is a difference monomial ordering
of $R$ that we call {\em clock-based}.

From now on, we assume that $R$ is endowed with a difference monomial ordering.
Let $f = \sum_i c_i m_i\in R$ with $m_i\in M$ and $0\neq c_i\in \KK$.
If $m_k = \max_\prec\{m_i\}$, we put $\lm(f) = m_k, \lc(f) = c_k$
and $\lt(f) = c_k m_k$. Since $\prec$ is a difference ordering, one has that
$\lm(\sigma(f)) = \sigma(\lm(f))$ and hence $\lc(\sigma(f)) = \lc(f),
\lt(\sigma(f)) = \sigma(\lt(f))$. If $G\subset R$, we denote
$\lm(G) = \{\lm(f) \mid f\in G,f\neq 0\}$ and we put $\LM(G) = \langle \lm(G) \rangle$.
Let $I$ be an ideal of $R$. A polynomial $f = \sum_i c_i m_i\in R$ is called
{\em normal modulo $I$} if $m_i\notin\LM(I)$, for all $i$. Since $R$ is endowed
with a monomial ordering, by a reduction process (see \cite{GLS,LS}) one proves that
for each polynomial $f\in R$ there is a unique element $\NF_I(f)\in R$ such that
$f - \NF_I(f)\in I$ and $\NF_I(f)$ is a normal polynomial modulo $I$.
In other words, the cosets of the normal monomials modulo $I$ form a $\KK$-linear
basis of the quotient algebra $R/I$. We call the polynomial $\NF_I(f)$
the {\em normal form of $f$ modulo $I$}.

\begin{proposition}
Let $G\subset R$. Then $\lm(\Sigma(G)) = \Sigma(\lm(G))$.
In particular, if $I$ is a difference ideal of $R$ then $\LM(I)$
is also a difference ideal.
\end{proposition}

\begin{proof}
Since $R$ is endowed with a difference monomial ordering, one has that
$\lm(\sigma(f)) = \sigma(\lm(f))$, for any $f\in R,f\neq 0$.
Then, $\Sigma(\lm(I)) = \lm(\Sigma(I))\subset \lm(I)$ and therefore
$\LM(I) = \langle \lm(I) \rangle$ is a difference ideal.
\end{proof}

\begin{definition}
Let $I\subset R$ be a difference ideal and $G\subset I$. We call $G$
a {\em difference \Gr\ basis} of $I$ if $\lm(G)$ is a difference basis
of $\LM(I)$. In other words, $\lm(\Sigma(G)) = \Sigma(\lm(G))$ is a basis
of $\LM(I)$, that is, $\Sigma(G)$ is a \Gr\ basis of $I$ as an ideal of $R$.
\end{definition}

For more details about difference \Gr\ bases and an optimized version
of the Buchberger procedure for these bases, we refer to \cite{GLS,LS}.

\begin{proposition}
Consider an explicit difference system $(\ref{dsys})$ and assume that
$R$ is endowed with a difference monomial ordering such that $x_i(r_i)\succ \lm(f_i)$,
for all $1\leq i\leq n$. Then, the set $G = \{x_1(r_1) - f_1,\ldots,x_n(r_n) - f_n\}$
is a difference \Gr\ basis.
\end{proposition}

\begin{proof}
From the assumption on the monomial ordering it follows that
$x_i(r_i + t) = \lm(x_i(r_i + t) - \sigma^t(f_i))$, for any $1\leq i\leq n$
and $t\geq 0$. By the linearity of these distinct leading monomials and
the Buchberger's Product Criterion (see, for instance, \cite{GP}) we conclude that
$\Sigma(G)$ is a \Gr\ basis, that is, $G$ is a difference \Gr\ basis.
\end{proof}

From now on, we assume that $x_i(r_i)\succ \lm(f_i)$, for any $i$.
If $I\subset R$ is the difference ideal generated by the set
$G = \{x_1(r_1) - f_1,\ldots,x_n(r_n) - f_n\}$, the above result
implies that $\LM(I) = \langle x_1(r_1),\ldots,x_n(r_n) \rangle_\sigma\subset R$.
In other words, the set of normal polynomials modulo $I$ is exactly the subalgebra
$\bR = \KK[\bX]$ where by definition
$\bX = \{x_1(0),\ldots,x_1(r_1-1),\ldots,x_n(0),\ldots,x_n(r_n-1)\}$.

\begin{proposition}
The map $\eta:R\to \bR, f\mapsto \NF_I(f)$ is an algebra homomorphism.
In other words, one has the algebra isomorphism
$\eta': R/I\to \bR$ such that $f + I\mapsto \NF_I(f)$.
\end{proposition}

\begin{proof}
By definition, we have that $\eta$ is a surjective $\KK$-linear map and
$\Ker\eta = I$. Then, it is sufficient to show that $m m'\notin\LM(I)$, for all
monomials $m,m'\notin\LM(I)$. This holds because $\LM(I)$ is an ideal which is
generated by variables.
\end{proof}

\begin{theorem}
\label{main}
Let $\bT:\bR\to \bR$ be the state transition endomorphism 
of the system $(\ref{dsys})$ and consider the algebra endomorphism $\sigma':R/I\to R/I$
such that $f + I\mapsto \sigma(f) + I$. Then, one has that $\bT \eta' = \eta' \sigma'$.
In particular, for each polynomial $f\in \bR$ and for all $t\geq 0$,
we have that $\bT^t(f) = \NF_I(\sigma^t(f))$.
\end{theorem}

\begin{proof}
Consider a polynomial $f\in \bR$, that is, $f = \NF_I(f)$.
The polynomial $\bT(f)\in \bR$ is obtained from the polynomial
$\sigma(f)\in R$ simply by applying the identities $x_i(r_i) = f_i$ ($1\leq i\leq n$).
Because $x_i(r_i) - f_i\in I$, we conclude that $\sigma(f) - \bT(f)\in I$.
\end{proof}

Observe finally that the above result implies that $f_{i,t} = \bT^t(x_i(0)) =
\NF_I(x_i(t))$.

\section{Invertible systems}

An important class of explicit difference systems are the ones such that
a $t$-state can be obtained from a $t'$-state also for $t'\geq t$.

\begin{definition}
For an explicit difference system $(\ref{dsys})$, consider the state transition
endomorphism $\bT:\bR\to\bR$ and the corresponding state transition map
$\T:\KK^r\to\KK^r$ $(r = r_1 + \ldots + r_n)$. We call the system {\em invertible}
if $\bT$ is an automorphism. In this case, $\T$ is also a bijective map. 
\end{definition}

We state now an algorithmic method to establish if an endomorphism of a polynomial
algebra is invertible and to compute its inverse. This important result is due
to Arno van den Essen (see \cite{VDE}, Theorem 3.2.1). Recall that a \Gr\ basis
$G = \{g_1,\ldots,g_r\}$ is called {\em (completely) reduced} if the polynomial
$g_i$ is normal modulo the ideal generated by $G\setminus\{g_i\}$, for all
$1\leq i\leq r$.

\begin{theorem}
\label{invth}
Let $X = \{x_1,\ldots,x_r\}, X' = \{x'_1,\ldots,x'_r\}$ be two disjoint
variable sets and define the polynomial algebras $P = \KK[X], P' = \KK[X']$
and $Q = \KK[X\cup X'] = P\otimes P'$. Consider an algebra endomorphism
$\varphi:P\to P$ such that $x_1\mapsto g_1,\ldots,x_r\mapsto g_r$ $(g_i\in P)$
and the corresponding ideal $J\subset Q$ which is generated by the set
$\{x'_1 - g_1,\ldots, x'_r - g_r\}$. Moreover, we endow the polynomial
algebra $Q$ by a product monomial ordering such that $X\succ X'$. Then,
the map $\varphi$ is an automorphism of $P$ if and only if the reduced
\Gr\ basis of $J$ is of the kind $\{x_1 - g'_1,\ldots,x_r - g'_r\}$
where $g'_i\in P'$, for all $1\leq i\leq r$. In this case, if $\varphi':P'\to P'$
is the algebra endomorphism such that $x'_1\mapsto g'_1, \ldots, x'_r\mapsto g'_r$
and $\xi: P\to P'$ is the isomorphism $x_1\mapsto x'_1, \ldots, x_r\mapsto x'_r$,
we have that $\xi\, \varphi^{-1}  = \varphi'\, \xi$.
\end{theorem}

For the context of explicit difference systems, the above criterion implies
the following results.

\begin{corollary}
\label{dinvth}
Let $\bT:\bR\to\bR$ be the state transition automorphism corresponding to an invertible
system $(\ref{dsys})$, namely $(1\leq i\leq n)$
\[
x_i(0)\mapsto x_i(1),\ldots,x_i(r_i-2)\mapsto x_i(r_i-1),
x_i(r_i-1)\mapsto f_i.
\]
Denote $\bR' = \KK[\bX']$ where
\[
\bX' =
\{x'_1(0),\ldots,x'_1(r_1-1),\ldots,x'_n(0),\ldots,x'_n(r_n-1)\}
\]
and put $Q = \bR\otimes \bR'$.
Consider the ideal $J\subset Q$ that is generated by the following polynomials,
for any $i = 1,2,\ldots,n$
\[
x'_i(0) - x_i(1),\ldots,x'_i(r_i-2) - x_i(r_i-1),x'_i(r_i-1) - f_i.
\]
With respect to a product monomial ordering of the algebra $Q$
such that $\bX\succ \bX'$, the reduced \Gr\ basis of $J$
has the following form
\[
x_i(1) - x'_i(0),\ldots,x_i(r_i-1) - x'_i(r_i-2), x_i(0) - f'_i
\]
where $f'_i\in \bR'$, for all $1\leq i\leq n$.
\end{corollary}

\begin{proof}
With respect to the considered monomial ordering of $Q$, we have clearly that
the set
\[
G = \bigcup_i \{x_i(1) - x'_i(0),\ldots,x_i(r_i-1) - x'_i(r_i-2)\}\subset J
\]
is the reduced \Gr\ basis of the ideal of $Q$ that is generated by it.
Since $\bT$ is an automorphism, by the Theorem \ref{invth} we conclude that
there are polynomials $f'_i\in \bR'$ ($1\leq i\leq n$) such that the set
\[
G\cup \bigcup_i \{x_i(0)-f'_i\}
\]
is the reduced \Gr\ basis of the ideal $J\subset Q$.
\end{proof}

By the above result, we obtain a sufficient condition to invertibility which is
immediate to verify.

\begin{corollary}
\label{dinvco}
Consider an explicit difference system $(\ref{dsys})$ and assume that
$f_i = x_{k_i}(0) + g_i$ where $\{x(0)_{k_1},\ldots,x(0)_ {k_n}\} = X(0) =
\{x_1(0),\ldots,x_n(0)\}$ and the polynomial $g_i$ has all variables
in the set $\bX\setminus X(0)$, for all $1\leq i\leq n$.
Then, the system $(\ref{dsys})$ is invertible.
\end{corollary}

\begin{proof}
With the same notations of Corollary \ref{dinvth}, consider the set
\[
G = \bigcup_i \{x_i(1) - x'_i(0),\ldots,x_i(r_i-1) - x'_i(r_i-2)\}\subset J
\]
and assume that the algebra $Q = \bR\otimes\bR'$ is endowed with a product monomial
ordering such that $\bX\succ \bX'$. Since the variables of $g_i$ are in
$\bX\setminus X(0)$, the normal form $g'_i$ modulo the ideal generated by $G$
is a polynomial with variables in the set
\[
\bX'\setminus \{x'_1(r_1-1),\ldots,x'_n(r_n-1)\}.
\]
Then, the reduced \Gr\ basis of the ideal $J\subset Q$ is given by the following
polynomials, for any $i = 1,2,\ldots,n$
\[
x_i(1) - x'_i(0),\ldots,x_i(r_i-1) - x'_i(r_i-2), x_{k_i}(0) - x'_i(r_i-1) - g'_i.
\]
By Theorem \ref{invth} we conclude that $\bT$ is an automorphism, that is, $(\ref{dsys})$
is an invertible system.
\end{proof}

\begin{definition}
For an explicit difference system $(\ref{dsys})$, consider the ideal
$J\subset Q = \bR\otimes \bR'$ which is generated by the following polynomials,
for each $i = 1,2,\ldots,n$
\[
x'_i(0) - x_i(1),\ldots,x'_i(r_i-2) - x_i(r_i-1),x'_i(r_i-1) - f_i.
\]
We call $J$ the {\em state transition ideal of the system (\ref{dsys})}.
\end{definition}

From now on, we assume that $Q$ is endowed with a product monomial ordering
such that $\bX\succ \bX'$.

\begin{definition}
\label{invsys}
Consider an invertible system $(\ref{dsys})$ and the corresponding state
transition ideal $J\subset Q$. If the set
\[
G = \bigcup_i
\{x_i(1) - x'_i(0),\ldots,x_i(r_i-1) - x'_i(r_i-2), x_i(0) - f'_i\}
\]
is the reduced \Gr\ basis of $J$, we denote by $g_i$ the image of $f'_i$
under the algebra isomorphism $\bR'\to \bR$ such that, for any $i = 1,2,\ldots,n$
\[
x'_i(0)\mapsto x_i(r_i-1), x'_i(1)\mapsto x_i(r_i-2), \ldots,
x'_i(r_i-1)\mapsto x_i(0).
\]
The {\em inverse of an invertible system $(\ref{dsys})$} is by definition
the following explicit difference system
\begin{equation}
\label{invdsys}
\left\{
\begin{array}{ccc}
x_1(r_1) & = & g_1, \\
& \vdots \\
x_n(r_n) & = & g_n. \\
\end{array}
\right.
\end{equation}
\end{definition}

Let $\bT,\bS:\bR\to\bR$ be the state transition endomorphisms of an invertible
system (\ref{dsys}) and its inverse system (\ref{invdsys}), respectively. Denote by
$\xi:\bR\to\bR$ the algebra automorphism such that
\[
x_i(0)\mapsto x_i(r_i-1), x_i(1)\mapsto x_i(r_i-2),\ldots,
x_i(r_i-1)\mapsto x_i(0).
\]
By Theorem \ref{invth} and Corollary \ref{dinvth}, we have that
$\xi \bS = \bT^{-1} \xi$.

\begin{proposition}
\label{invstat}
Let $(\ref{invdsys})$ be the inverse system of an invertible system $(\ref{dsys})$.
If $(a_1,\ldots,a_n)$ is a $\KK$-solution of $(\ref{dsys})$, consider
its $t$-state $(t\geq 0)$
\[
v = (a_1(t),\ldots,a_1(t+r_1-1),\ldots,a_n(t),\ldots,a_n(t+r_n-1)).
\]
Denote by $(b_1,\ldots,b_n)$ the $\KK$-solution of $(\ref{invdsys})$ whose
initial state is
\[
v' = (a_1(t+r_1-1),\ldots,a_1(t),\ldots,a_n(t+r_n-1),\ldots,a_n(t)).
\]
If the $t$-state of $(b_1,\ldots,b_n)$ is
\[
u' = (b_1(t),\ldots,b_1(t+r_1-1),\ldots,b_n(t),\ldots,b_n(t+r_n-1)),
\]
then the initial state of $(a_1,\ldots,a_n)$ is
\[
u = (b_1(t+r_1-1),\ldots,b_1(t),\ldots,b_n(t+r_n-1),\ldots,b_n(t)).
\]
\end{proposition}

\begin{proof}
Denote by $\T,\S:\KK^r\to\KK^r$ ($r = r_1 + \ldots + r_n$) the state
transition maps of the systems (\ref{dsys}),(\ref{invdsys}), respectively.
By definition, we have that $u' = \S^t(v')$. Since $\xi \bS = \bT^{-1} \xi$,
we conclude that $u = \T^{-t}(v)$.
\end{proof}

Another useful notion is the following one.

\begin{definition}
An explicit difference system $(\ref{dsys})$ is called {\em reducible}
if there is an integer $0 < m < n$ such that we have a subsystem
\begin{equation}
\label{subdsys}
\left\{
\begin{array}{ccc}
x_1(r_1) & = & f_1, \\
& \vdots \\
x_m(r_m) & = & f_m. \\
\end{array}
\right.
\end{equation}
In other words, one has that $f_1,\ldots,f_m\in\bR_m = \KK[\bX_m]$ where
by definition $\bX_m = \{x_1(0),\ldots,x_1(r_1-1),\ldots,x_m(0),\ldots,x_m(r_m-1)\}$.
In this case, the state transition endomorphism and map of $(\ref{subdsys})$
are just the restrictions of the corresponding functions of $(\ref{dsys})$
to the subring $\bR_m\subset\bR$ and the subspace $\KK^k\subset \KK^r$
$(k = r_1 + \ldots + r_m, r = r_1 + \ldots + r_n)$, respectively.
\end{definition}

The following result is obtained immediately.

\begin{proposition}
Let $(\ref{dsys})$ be a reducible invertible system. Then, its subsystem 
$(\ref{subdsys})$ is also invertible. Moreover, the inverse system of $(\ref{dsys})$
is also reducible with a subsystem which is the inverse system of $(\ref{subdsys})$.
\end{proposition}

\section{Periodic systems}

\begin{definition}
\label{period}
For an invertible system $(\ref{dsys})$, consider the state transition
map $\T:\KK^r\to\KK^r$ $(r = r_1 + \ldots + r_n)$. We call the system
{\em periodic} if there is an integer $d > 0$ such that $\T^d = \id$.
In this case, the period of the map $\T$ is called the {\em period of the
system $(\ref{dsys})$}.
\end{definition}

\begin{proposition}
Consider a periodic system $(\ref{dsys})$ with period $d$. If $(a_1,\ldots,a_n)$
is a $\KK$-solution of $(\ref{dsys})$, then all functions $a_i$ are periodic, that is,
$a_i(t) = a_i(t + d)$ for all clocks $t\geq 0$.
\end{proposition}

\begin{proof}
If $v\in\KK^r$ is the initial state of $(a_1,\ldots,a_n)$, by the argument
of Theorem \ref{unique} we have that $a_1(t)$ is the first coordinate
of the vector $\T^t(v)\in\KK^r$. Since $\T^t = \T^{t + d}$, one has that
$\T^t(v) = \T^{t + d}(v)$ and therefore $a_1(t) = a_1(t + d)$. In a similar way,
we also prove that $a_i(t) = a_i(t + d)$ ($1 < i\leq n$).
\end{proof}

Note that if $\KK = \GF(q)$ is a finite field, the symmetric group
$\SS(\KK^r)$ has finite order and therefore all invertible systems
are in fact periodic. We also observe that if $\KK$ is an infinite field,
then the state transition endomorphism $\bT$ is bijective if and only if
the state transition map $\T$ is bijective. Moreover, we have that
$\bT$ is periodic if and only if $\T$ is periodic and in this case
these maps have the same period. Such facts are consequences
of the following general result (see for instance \cite{Ma, Re}).

\begin{proposition}
\label{corresp}
Consider a polynomial algebra $P = \KK[x_1,\ldots,x_r]$ and an algebra endomorphism
$\varphi:P\to P$ such that $x_1\mapsto g_1,\ldots,x_r\mapsto g_r$ $(g_i\in P)$.
Denote by $\hvarphi:\KK^r\to\KK^r$ the corresponding polynomial map, that is,
for any $(\alpha_1,\ldots,\alpha_r)\in\KK^r$
\[
(\alpha_1,\ldots,\alpha_r)\mapsto (g_1(\alpha_1,\ldots,\alpha_r),\ldots,
g_r(\alpha_1,\ldots,\alpha_r)).
\]
The map $\varphi\mapsto \hvarphi$ is a homomorphism from the monoid
of algebra endomorphisms of $P$ to the monoid of polynomial maps
$\KK^r\to\KK^r$. If $\KK$ is an infinite field, this monoid homomorphism
is bijective. Otherwise, if $\KK = \GF(q)$ then the map $\varphi\mapsto \hvarphi$
induces a monoid isomorphism from the monoid of algebra endomorphisms
of the quotient algebra $P/L$,
where $L = \langle x_1^q - x_1,\ldots,x_r^q - x_r\rangle\subset P$.
Note that $P$ and $P/L$ are the coordinate algebras of the affine space $\KK^r$
for the case that $\KK$ is an infinite or finite field, respectively.
\end{proposition}

An important and difficult task is to compute, or at least bound, the period
of a periodic explicit difference system. As usual, the task becomes easy
in the linear case.

\begin{definition}
An explicit difference system $(\ref{dsys})$ is called {\em linear} if all
polynomials $f_i$ $(1\leq i\leq n)$ are homogeneous linear ones. In other words,
the state transition map $\T:\KK^r\to\KK^r$ is a $\KK$-linear endomorphism
of the vector space $\KK^r$.
\end{definition}

Restating the Rational (or Frobenius) Canonical Form of a square matrix
(see, for instance, \cite{Ja}) in terms of $\KK$-linear endomorphisms,
one has the following result.

\begin{proposition}
\label{rcf}
Let $\psi:\KK^r\to\KK^r$ be any $\KK$-linear endomorphism. Then, there is a
$\KK$-linear automorphism $\xi:\KK^r\to\KK^r$ such that $\psi' = \xi \psi \xi^{-1}$
can be decomposed as a direct sum $\psi' = \bigoplus_{1\leq i\leq n} \psi'_i$
where $\psi'_i:\KK^{r_i}\to \KK^{r_i}$ $(r_1 + \ldots + r_n = r)$ is a $\KK$-linear
endomorphism such that, for any $(\alpha_0,\ldots,\alpha_{r_i-1})\in\KK^{r_i}$
\[
\psi'_i(\alpha_0,\ldots,\alpha_{r_i-2},\alpha_{r_i-1}) =
(\alpha_1,\ldots,\alpha_{r_i-1},g_i(\alpha_0,\ldots,\alpha_{r_i-1}))
\]
and $g_i$ is a homogeneous linear polynomial in $r_i$ variables.
It follows that if $\psi$ is an automorphism of finite period $d$, then
$d = \lcm(d_1,\ldots,d_n)$ where $d_i$ is the period of $\psi'_i$.
\end{proposition}

Note that the above result provides that, up to an invertible $\KK$-linear
change of variables, any linear difference system can be obtained
in a canonical form, say (\ref{dsys}), where $f_i$ ($1\leq i\leq n$) is
a linear form which is defined only over the set of variables $\{x_i(0),\ldots,x_i(r_i-1)\}$.
In other words, the system is the join of linear difference equations
on disjoint sets of variables. In Cryptography (see, for instance, \cite{Go,SP}),
a linear difference equation is called a {\em Linear Feedback Shift Register}
or briefly {\em LFSR}.

For cryptographic applications, to have a periodic difference system 
with a large period $d$ is a useful property. In the linear case, according to
Proposition \ref{rcf}, to maximize $d$ one needs that all $d_i$ are coprime
so that $d = d_1\cdots d_n$. Then, the problem reduces to maximize the period
of each single periodic linear difference equation. This problem has a well-known
solution (see \cite{Go}) when $\KK = \GF(p) = \ZZ_p$ with $p$ a prime number.
For the purpose of completeness, we provide this result.

\begin{proposition}
Consider an invertible linear difference equation
\begin{equation}
\label{deq}
x(r) = \sum_{0\leq i\leq r-1} c_i x(i)\ (c_i\in\ZZ_p).
\end{equation}
Denote $g = t^r - \sum_i c_i t^i\in\ZZ_p[t]$ and assume that $g$ is an irreducible
polynomial. Consider the finite field $\FF = \GF(p^r) = \ZZ_p[t]/(g)$ and
the corresponding multiplicative (cyclic) group $\FF^* = \FF\setminus\{0\}$. If the element
$\alpha = t + (g)\in \FF^*$ is a generator of $\FF^*$, that is,
$g$ is a {\em primitive polynomial of $\ZZ_p[t]$}, then the equation $(\ref{deq})$
has maximal period $p^r-1$.
\end{proposition}

\begin{proof}
The matrix corresponding to the state transition $\KK$-linear map
$\T:\ZZ_p^r\to\ZZ_p^r$ of (\ref{deq}) with respect to the canonical basis
of $\ZZ_p^r$, is indeed the companion matrix $A$ of the polynomial $g$, that is
\[
A =
\left(
\begin{array}{ccccc}
0 & 1 & 0 & \ldots & 0 \\
0 & 0 & 1 & \ldots & 0 \\
\vdots & \vdots & \vdots & \ddots & \vdots \\
0 & 0 & 0 & \ldots & 1 \\
c_0 & c_1 & c_2 & \ldots & c_{r-1} \\
\end{array}
\right).
\]
Consider the algebra $\M_r(\ZZ_p)$ of all square matrices of order $r$ with entries
in the field $\ZZ_p$ and denote by $\ZZ_p[A]\subset \M_r(\ZZ_p)$ the subalgebra
that is generated by the matrix $A$. It is well-known that the minimal polynomial
of the companion matrix $A$ is exactly $g$ and hence $\FF = \ZZ_p[\alpha]$
is isomorphic to $\ZZ_p[A]$ by the map $\alpha\mapsto A$. Since $\alpha$ is a generator,
that is, an element of maximal period in the cyclic group $\FF^*$, we obtain
that the period of $\alpha$ and $A$ is exactly $p^r - 1$.
\end{proof}

\section{Difference stream/block ciphers}

From now on, let $\KK = \GF(q)$ be a finite field. It is important to note
that in this case, by Lagrange interpolation, any function $\KK^r\to \KK$ is in fact
a polynomial one. Moreover, if $f\in\KK[x_1,\ldots,x_r]$ is the corresponding
polynomial, we can assume that $f$ is normal modulo the ideal
$L = \langle x_1^q - x_1,\ldots,x_r^q - x_r \rangle$, that is, all exponents
in the monomials of $f$ are strictly less than $q$.

\begin{definition}
A {\em difference stream cipher} $\cC$ is an explicit difference system $(\ref{dsys})$
together with a polynomial $f\in \bR$.
If $(a_1,\ldots,a_n)$ is a $\KK$-solution of $(\ref{dsys})$, its initial state
is called the {\em key of $(a_1,\ldots,a_n)$}. Moreover, if $v(t)\in\KK^r$
$(r = r_1 + \ldots + r_n)$ is the $t$-state of $(a_1,\ldots,a_n)$, the function
$b:\NN\to\KK$ such that $b(t) = f(v(t))$ for all $t\geq 0$, is called
the {\em keystream of $(a_1,\ldots,a_n)$}. Finally, we call $f$
the {\em keystream polynomial of the cipher $\cC$}.
\end{definition}

If $(\ref{dsys})$ is linear, that is, a system of LFSRs, the polynomial $f$ is required
non-linear and it is usually called a {\em combining or filtering function}.
Observe that a difference stream cipher can also be defined as a special explicit
difference system
\[
\left\{
\begin{array}{ccc}
x_1(r_1) & = & f_1, \\
& \vdots \\
x_n(r_n) & = & f_n, \\
y(0) & = & f.
\end{array}
\right.
\]
In fact, by a $\KK$-solution $(a_1,\ldots,a_n,b)$ of such a system
one obtains the keystream function $b:\NN\to\KK$ of the $\KK$-solution
$(a_1,\ldots,a_n)$ of $(\ref{dsys})$.

In Cryptography, a stream cipher (see, for instance, \cite{Kl}) operates simply
by adding and subtracting the keystream to a {\em stream of plaintexts or
ciphertexts}. Such a stream is by definition a function $\NN\to\KK$.
By a {\em known plaintext attack} we can assume the knowledge of the keystream
as the difference between the known ciphertext and plaintext streams.
Note that the keystream is usually provided by a stream cipher after
a sufficiently high number of clocks in order to prevent cryptanalysis.

\begin{definition}
\label{stream}
Let $\cC$ be a difference stream cipher consisting of the system $(\ref{dsys})$
and the keystream polynomial $f$. Let $b:\NN\to\KK$ be the keystream of a
$\KK$-solution of $(\ref{dsys})$ and fix an integer $T\geq 0$. Consider the ideal
\[
J = \sum_{t\geq T} \langle \sigma^t(f) - b(t) \rangle\subset R
\]
and denote by $V_\KK(J)$ the set of simultaneous $\KK$-solutions of all
polynomials in $J$, or equivalently, of its generators. An {\em algebraic attack
to $\cC$ by the keystream $b$ after $T$ clocks} consists in computing
the $\KK$-solutions $(a_1,\ldots,a_n)$ of the system $(\ref{dsys})$ such that
$(a_1,\ldots,a_n)\in V_\KK(J)$. In other words, if we consider the difference
ideal corresponding to $(\ref{dsys})$, that is,
$I = \langle x_1(r_1) - f_1,\ldots,x_n(r_n) - f_n \rangle_\sigma\subset R$
then one wants to compute $V_\KK(I + J) = V_\KK(I)\cap V_\KK(J)$.
\end{definition}

Since the given function $b$ is the keystream of a $\KK$-solution of $(\ref{dsys})$,
say $(a_1,\ldots,a_n)$, we have that $(a_1,\ldots,a_n)\in
V_\KK(I + J)\neq\emptyset$. For actual ciphers, we generally have that
$V_\KK(I + J) = \{(a_1,\ldots,a_n)\}$. 

\begin{definition}
With the notation of Definition \ref{stream}, denote by $\bV_\KK(I + J)\subset \KK^r$
the set of keys, that is, initial states of the $\KK$-solutions
$(a_1,\ldots,a_n)\in V_\KK(I + J)$. By Theorem \ref{unique}, there is a bijective
map $V_\KK(I + J)\to \bV_\KK(I + J)$ and we have that $\bV_\KK(I) = \KK^r$.
\end{definition}

\begin{theorem}
\label{keyeq}
Let $\bT:\bR\to \bR$ be the state transition
endomorphism of the system $(\ref{dsys})$ and put $f'_t = \bT^t(f)\in \bR$,
for all $t\geq 0$. Moreover, define the ideal
\[
J' = \sum_{t\geq T} \langle f'_t - b(t) \rangle\subset \bR.
\]
Then, we have that $\bV_\KK(I + J) = V_\KK(J')$.
\end{theorem}

\begin{proof}
Let $(a_1,\ldots,a_n)$ be a $\KK$-solution of (\ref{dsys}) and denote by $v(t)$
its $t$-state ($t\geq 0$). By the identity (\ref{fund}), one obtains that
\[
f'_t(v(0)) = \bT^t(f)(v(0)) = f(\T^t(v(0))) = f(v(t)).
\]
We conclude that the condition $f(v(t)) = b(t)$ ($t\geq T$) is equivalent to
the condition $f'_t(v(0)) = b(t)$.
\end{proof}

By assuming that $x_i(r_i)\succ \lm(f_i)$ for all $1\leq i\leq n$, that is,
$G = \{x_1(r_1) - f_1, \ldots, x_n(r_n) - f_n\}$ is a difference \Gr\ basis
of the difference ideal $I = \langle G \rangle_\sigma$, by Theorem \ref{main}
one obtains that
\[
f'_t = \bT^t(f) = \NF_I(\sigma^t(f)).
\]

In actual algebraic attacks, we are given a finite number of values of the
keystream $b$, that is, for a fixed integer bound $B\geq T$, we consider
the finitely generated ideal
\[
J'_B = \sum_{T\leq t\leq B} \langle f'_t - b(t) \rangle\subset \bR.
\]
We have that $J' = \bigcup_{B\geq T} J'_B$ where $J'_B\subset J'_{B+1}$.
Since the polynomial algebra $\bR$ is finitely generated and hence
Noetherian, one has that $J'_B = J'$ for some $B\geq T$. In other words, we don't
lose any equation satisfied by the keys if a sufficiently large number of keystream
values is provided for the attack.

To compute the set $V_\KK(J'_B)\subset \KK^r$ one can use essentially \Gr\ bases
or SAT solvers when $\KK = \GF(2)$ (see, for instance, \cite{Ba}).
For real ciphers, one generally has that $V_\KK(J'_B) = V_\KK(J')$
contains a single $\KK$-solution, that is, a single key. The Nullstellensatz
over finite fields (see, for instance \cite{Gh}) implies the following result.

\begin{proposition}
\label{uniqueGB}
Let $\KK = \GF(q)$ be a finite field and consider the polynomial algebra
$P = \KK[x_1,\dots,x_r]$ and the ideal
$L = \langle x_1^q - x_1, \ldots, x_r^q - x_r \rangle\subset P$.
Moreover, let $J\subset P$ be any ideal and denote by $V(J)$ the set
of $\bar{\KK}$-solutions of all polynomials $f\in J$ where the field $\bar{\KK}$
is the algebraic closure of $\KK$. We have that $V(L) = \KK^r$
and $V_\KK(J) = V(J)\cap \KK^r = V(J + L)$ where $J + L\subset P$ is a radical
ideal. Moreover, if $V_\KK(J) = \{(\alpha_1,\ldots,\alpha_r)\}$ then
$G = \{x_1 - \alpha_1,\ldots,x_r - \alpha_r\}$ is the (reduced) universal
\Gr\ basis of $J + L$, that is, its \Gr\ basis with respect to all monomial
orderings of $P$.
\end{proposition}

The above result is very useful for the algebraic attacks because \Gr\ bases
computations are very sensitive to the monomial orderings and we are free
here to choose the most efficient orderings such as DegRevLex. Another possible
optimization when performing the Buchberger algorithm on the ideal $J + L\subset P$
consists in skipping all remaining S-polynomials once each variable $x_i$
($1\leq i\leq r$) has been obtained as the leading monomial of an element
in the current \Gr\ basis. In some cases, this trick speeds up the computation
in a significant way.

Let $\cC$ be a difference stream cipher which is given by the system $(\ref{dsys})$
and the keystream polynomial $f$. The polynomial $f'_t = \bT^t(f)\in \bR$ generally
has a high degree if $T$ is large with respect to $0$. This is usually
the case in actual ciphers where a high number of clocks is required before
the keystream appears. Nevertheless, if the system $(\ref{dsys})$ is invertible
we can always assume that $T = 0$. In fact, by means of the notion of inverse
system in Definition \ref{invsys}, to compute the $T$-state is completely
equivalent to compute the initial state, that is, the key of a $\KK$-solution
of $(\ref{dsys})$. This is a very effective optimization because it drastically
reduces the degrees of the generators of the ideal
$J'_B = \sum_{T\leq t\leq B} \langle f'_t - b(t) \rangle$
to those of the generators of the ideal
$J''_B = \sum_{0\leq t\leq B-T} \langle f'_t - b(T+t) \rangle$.
Recall that we have to compute a \Gr\ basis for obtaining $\KK$-solutions
and such computations are very sensitive to the degree of the generators.
We apply this strategy when attacking the stream cipher \Biv\ in Section 6.
If the polynomials $f'_t$ have still high degrees and they are even difficult
to compute, an alternative strategy consists in computing directly
the $\KK$-solutions of the system $(\ref{dsys})$ which are also solutions
of the fixed degree polynomials $\sigma^t(f) - b(T+t)$ $(0\leq t\leq B-T)$.
Even though the clocks of the variables in $X$ can be bounded, this strategy
has the main drawback that one has to compute a \Gr\ basis over a generally
high number of variables.

We have just observed that difference stream ciphers that are defined by invertible
systems show some lack of security with respect to algebraic attacks. On the other
hand, invertible systems can be used to define block ciphers.

\begin{definition}
A {\em difference block cipher} $\cC$ is a reducible invertible system
$(\ref{dsys})$ together with an integer $T\geq 0$. If $(\ref{subdsys})$ is
the subsystem of $(\ref{dsys})$, we put $k = r_1 + \ldots + r_m$ and
$l = r_{m+1} + \ldots + r_n$. If a $t$-state of a $\KK$-solution $(a_1,\ldots,a_n)$
of $(\ref{dsys})$ is denoted as the pair $(u(t),v(t))\in\KK^k\times\KK^l = \KK^r$,
we call $u(0)$ the {\em key}, $v(0)$ the {\em plaintext} and $v(T)$
the {\em ciphertext of $(a_1,\ldots,a_n)$}. Moreover, we call $(u(T),v(T))$
the {\em final state of $(a_1,\ldots,a_n)$} and $(\ref{subdsys})$
the {\em key subsystem of the cipher $\cC$}.
\end{definition}

With the language of Cryptography, the {\em encryption function}
$E_{u(0)}:\KK^l\to\KK^l$ of the difference block cipher $\cC$ is given by
the map $v(0)\mapsto v(T)$, where the pair $(u(0),v(0))$ varies in the affine
space $\KK^k\times\KK^l$ of all initial states of the $\KK$-solutions
of $(\ref{dsys})$. To provide the decryption function we introduce the following
notion.

\begin{definition}
Let $\cC$ be a difference block cipher consisting of a reducible invertible
system $(\ref{dsys})$ and a clock $T\geq 0$. The {\em inverse cipher of $\cC$}
is by definition the inverse system of $(\ref{dsys})$ together with $T$.
\end{definition}

Let $\cC'$ be the inverse cipher of $\cC$ where $(\ref{invdsys})$ is the inverse
system of $(\ref{dsys})$. Consider also the key subsystem $(\ref{subdsys})$ of $\cC$.
If $u(0)$ is the key of a solution of $(a_1,\ldots,a_n)$ of $(\ref{dsys})$,
we can compute $u(T)$ by means of $(\ref{subdsys})$ without knowing $v(0)$.
If we are given the ciphertext $v(T)$, we have hence the final state $(u(T),v(T))$
of $(a_1,\ldots,a_n)$. By Proposition \ref{invstat}, the inverse system
$(\ref{invdsys})$ provides the computation of the initial state $(u(0),v(0))$
of $(a_1,\ldots,a_n)$ and in particular of the plaintext $v(0)$.
In other words, the decryption function $D_{u(0)}:\KK^l\to\KK^l$ is obtained as
the map $v(T)\mapsto v(0)$ which is computable by means of the systems $(\ref{invdsys}),
(\ref{subdsys})$.

\begin{definition}
\label{block}
Let $\cC$ be a difference block cipher given by a reducible invertible system
$(\ref{dsys})$ and a clock $T\geq 0$. For all $t\geq 0$,
let $(u(t),v(t))\in\KK^k\times\KK^l$ be the $t$-state of a $\KK$-solution
of $(\ref{dsys})$ where we denote
\[
v(t) =
(a_{m+1}(t),\ldots,a_{m+1}(t+r_{m+1}-1),\ldots,a_n(t),\ldots,a_n(t+r_n-1)).
\]
Consider the corresponding linear ideal
\[
J(t) = \sum_{m+1\leq i\leq n}
\langle x_i(t) - a_i(t),\ldots, x_i(t+r_i-1) - a_i(t+r_i-1) \rangle\subset R
\]
and put $J = J(0) + J(T)$. An {\em algebraic attack to $\cC$ by the
plaintext-ciphertext pair $(v(0),v(T))$} consists in computing the $\KK$-solutions
$(a_1,\ldots,a_n)$ of the system $(\ref{dsys})$ such that
$(a_1,\ldots,a_n)\in V_\KK(J)$. If $I =
\langle x_1(r_1) - f_1,\ldots,x_n(r_n) - f_n \rangle_\sigma\subset R$,
this is equivalent to compute $V_\KK(I + J) = V_\KK(I)\cap V_\KK(J)$.
\end{definition}

Note that the above attack belongs to the class of known plaintext attacks.
Since the given pair $(v(0),v(T))$ is obtained by the states $(u(t),v(t))$
of a $\KK$-solution of $(\ref{dsys})$, say $(a_1,\ldots,a_n)$, we have that
$(a_1,\ldots,a_n)\in V_\KK(I + J)\neq\emptyset$. For actual ciphers,
one generally has that the set $V_\KK(I + J)$ contains more than one $\KK$-solution.
Since computing a unique solution by a single DegRevLex-\Gr\ basis can be faster
than calculating multiple solutions via conversion to an elimination ordering
(FGLM algorithm \cite{FGLM}), we prefer to obtain uniqueness by attacking with multiple
plaintext-ciphertext pairs.

Precisely, fix an integer $s > 1$ and let $(u(t),v^{(i)}(t))$ ($1\leq i\leq s$)
be the $t$-state of a $\KK$-solution $(a_1,\ldots,a_m,a^{(i)}_{m+1},\ldots,a^{(i)}_n)$
of the system (\ref{dsys}) where $(a_1,\ldots,a_m)$ is some fixed $\KK$-solution
of the key subsystem (\ref{subdsys}). In other words, we consider some
plaintext-ciphertext pairs $(v^{(i)}(0),v^{(i)}(T))$ ($1\leq i\leq s$)
which are obtained by the same key $u(0)$. To describe properly a multiple
pairs attack, we also need the following notations.

For each $i = 1,2,\ldots,s$ and $t\geq 0$, consider the set of variables
$X^{(i)}(t) = \{x_1(t),\ldots,x_m(t), x^{(i)}_{m+1}(t),\ldots,x^{(i)}_n(t)\}$
where $\bigcap_i X^{(i)}(t) = \{x_1(t),\ldots,x_m(t)\}$. We put
$X^{(i)} = \bigcup_{t\geq 0} X^{(i)}(t)$ and $R^{(i)} = \KK[X^{(i)}]$.
The polynomial algebra $R^{(i)}$ is clearly isomorphic to $R$ and we denote by
$I^{(i)}\subset R^{(i)}$ the ideal which is isomorphic to $I\subset R$.
For all $1\leq i\leq s$, consider (\ref{dsys}) as written
in the variables $X^{(i)}$ and let $J^{(i)}(t)\subset R^{(i)}$ ($t\geq 0$)
be the linear ideal corresponding to the $t$-state of the $\KK$-solution
$(a_1,\ldots,a_m,a^{(i)}_{m+1},\ldots,a^{(i)}_n)$ of the system (\ref{dsys}).
Moreover, we put $X' = \bigcup_i X^{(i)}, R' = \KK[X']$ and we denote by
$I',J'(t)$ the ideals of $R'$ that are generated by $\sum_i I^{(i)},
\sum_i J^{(i)}(t)$, respectively. Finally, we put $J' = J'(0) + J'(T)$.

\begin{definition}
Let $\cC$ be a difference block cipher given by a reducible invertible system
$(\ref{dsys})$ and a clock $T\geq 0$. An {\em algebraic attack to $\cC$
by the multiple plaintext-ciphertext pairs $(v^{(i)}(0),v^{(i)}(T))$
($1\leq i\leq s$)} consists in computing the $\KK$-solutions
\[
(a_1,\ldots,a_m,a^{(1)}_{m+1},\ldots,a^{(s)}_{m+1},\ldots,
a^{(1)}_n,\ldots,a^{(s)}_n)
\in V_\KK(I' + J').
\]
\end{definition}

For real ciphers, a sufficiently large number of pairs implies that
we have $V_\KK(I' + J') = \{(a_1,\ldots,a_m,a^{(1)}_{m+1},\ldots,a^{(s)}_{m+1},
\ldots,a^{(1)}_n,\ldots,a^{(s)}_n)\}$. By bounding the clocks of the variables
in $X'$, one can compute this unique $\KK$-solution and hence its key
$u(0) = (a_1(0),\ldots,a_1(r_1-1),\ldots,a_m(0),\ldots,a_m(r_m-1))$
using a \Gr\ basis computation as in Proposition \ref{uniqueGB}.
Alternatively, for $\KK = \GF(2)$ one can use SAT solvers or other methods.
In Section 7 we make use of multiple pairs when attacking the block cipher \Kee.

Since the final clock $T$ is usually chosen a large one, the main drawback
of this approach is the high number of variables. Despite this, such a strategy
generally appears to be more viable than elimination techniques. Indeed, the normal
forms of the generators of $J'(T)$ modulo $I' + J'(0)$ belong to $\bR_m$
but they may have very high degrees because of the large clock $T$. As for stream
ciphers, the main problem is hence to reduce somehow the final clock $T$.
Even though the system (\ref{dsys}) of the block cipher $\cC$ is invertible,
note that we cannot attack an internal state instead of the initial one because
the set $V_\KK(I' + J'(t))$ ($t\leq T$) generally contains too many solutions.
In other words, a ciphertext only attack is generally too weak for difference block ciphers.

A better strategy is possible when the period of the key subsystem (\ref{subdsys}),
say $d$, is sufficiently small. This technique has been introduced in \cite{CBB,CBW}
to attack \Kee. To simplify its description, let us assume that the final state $T$
is a multiple of $d$. Consider the state transition map $\T:\KK^r\to\KK^r$
of the explicit difference system (\ref{dsys}) and denote by 
$\S:\KK^k\to\KK^k$ the state transition map of the subsystem (\ref{subdsys}).
Recall that $\S$ is just the restriction of the map $\T$ to the subspace
$\KK^k\subset\KK^r$ ($k = r_1 + \ldots + r_m, r = r_1 + \ldots + r_n$).
By definition of period, one has that $\S^d = \id$ and therefore $\S^d(u) = u$,
for all $u\in\KK^k$. For any $t\geq 0$, denote by $(u(t),v(t))\in\KK^k\times\KK^l = \KK^r$
the $t$-state of a $\KK$-solution $(a_1,\ldots,a_n)$ of the system (\ref{dsys}).
Then, the encryption function $E_{u(0)}:\KK^l\to\KK^l$ corresponding to the key
$u(0)\in\KK^k$ is the map $v(0)\mapsto v(T)$. By a chosen plaintext attack,
we can assume the knowledge of the bijection $E_{u(0)}$. If $\KK^l$ is a large space
($l\to\infty$), one has a probability equal to $1 - 1/e\approx 0.63$ (see \cite{CBB},
Section 4.1) that $E_{u(0)}$ has one or more fixed points $v(0) = v(T)$. Observe now that
$v(0) = v(d)$ implies that $v(0) = v(T)$. In fact, by definition $(u(d), v(d)) =
\T^d(u(0),v(0))$ and we have that $u(d) = \S^d(u(0)) = u(0)$. Then, from $v(0) = v(d)$
it follows that $(u(0), v(0)) = \T^d(u(0),v(0))$ and hence $(u(0), v(0)) =
\T^T(u(0),v(0))$ because $T$ is a multiple of $d$. We conclude that among
the fixed points of the encryption function $E_{u(0)}$ one has the fixed points
of the map $v(0)\mapsto v(d)$. If $v(0) = v(d)$ is such a fixed point, we have that
$(v(0),v(0)) = (v(0),v(d)) = (v(0),v(T))$, that is, $(v(0),v(0))$ is a
plaintext-ciphertext pair for the final clocks $d$ and $T$. In other words,
by means of such pairs we can perform an algebraic attack to the difference block
cipher $\cC$ assuming that the final clock is just the period of the key subsystem.

Let us conclude this section with a final general observation. When we apply Proposition
\ref{uniqueGB} for solving polynomial systems, an essential trick consists in adding some
linear polynomials to the considered ideal $J + L\subset \KK[x_1,\dots,x_r]$
($L = \langle x_1^q - x_1, \ldots, x_r^q - x_r \rangle$) in order to speed up
the \Gr\ basis computation. Such linear polynomials are either elements of $J$ which
are given or computed ones, or they correspond to the evaluations of some subset
of variables $\{x_{i_1},\ldots,x_{i_s}\}\subset\{x_1,\ldots,x_r\}$. If some of these
evaluations, say $x_{i_k} = \alpha_{i_k}$ ($\alpha_{i_k}\in\KK$), is wrong and
$V_\KK(J)$ contains a unique solution, one has that
\[
J + L + \langle x_{i_1} - \alpha_{i_1}, \ldots, x_{i_s} - \alpha_{i_s} \rangle =
\langle 1 \rangle
\]
and the \Gr\ basis computation stops as soon as the element 1 is obtained.
Note that using instead a SAT solver, the answer ``UNSAT'' essentially arrives
when the full space $\KK^r$ ($\KK = \GF(2)$) has been examined. This means that
for wrong evaluations, which are all except that one, \Gr\ basis solving is generally
faster than SAT solving. We have evidence of this in practice in Section 6.

Note that solving after the evaluation of some bunch of variables is usually
called a {\em guess-and-determine strategy} (see, for instance, \cite{EVP})
or a {\em hybrid strategy} (see \cite{BFP}) in the case of (semi-)regular
polynomial systems. The latter case cannot be generally assumed for algebraic
attacks to difference ciphers because the polynomials we obtain by means
of elimination techniques such as Theorem \ref{keyeq} do not seem random
at all. For this reason, as in the paper \cite{EVP}, we prefer to consider
the experimental running time of a guess-and-determine strategy as the product
$a\cdot q^s$ where $a$ is the average solving time for a single guess
and $q^s$ is the number of guesses of $s$ variables when $\KK = \GF(q)$.
In other words, the complexity of such a strategy is $\cO(q^s)$
where $q^s$ is the total number of solving processes to be performed
in a reasonable time.

Of course, by assuming with probability $\geq 1/2$ that the correct guess 
is found in half of the space $\KK^q$, we obtain that the average running time
is reduced to $a\cdot q^s/2$. We conclude by observing that the choice
of the variables to be evaluated is a very important issue to optimize
a guess-and-determine strategy. The parallelization of the computation which
can be obtained simply by dividing the guess space in different subsets
is also a viable way to scale down the complexity.

\section{Attacking \Biv}

Aiming to illustrate by a concrete example how to perform in practice
an algebraic attack to a difference stream cipher as described in the previous
section, we start considering \Tri\ which is a well-known stream cipher
designed in 2003 by De Canni\`ere and Preneel as a submission to European
project eSTREAM \cite{DCP}. In fact, \Tri\ was one of the winners
of the project for the category of hardware-oriented ciphers. Even though
it has been widely cryptanalysed, no critical attacks are known up to date.
The system of explicit difference equations describing \Tri\ looks quite simple
since it consists only of three quadratic equations over the base field $\KK = \GF(2)$,
namely

\begin{equation}
\label{tri}
\left\{
\begin{array}{rcl}
x(93) & = & z(0) + x(24) + z(45) + z(1)z(2), \\
y(84) & = & x(0) + y(6) + x(27) + x(1)x(2), \\
z(111) & = & y(0) + y(15) + z(24) + y(1)y(2). \\
\end{array}
\right.
\end{equation}
Its keystream polynomial is a homogeneous linear one
\begin{equation*}
\begin{array}{rcl}
f & = & x(0) + x(27) + y(0) + y(15) + z(0) + z(45).
\end{array}
\end{equation*}
Therefore, a $t$-state consists of $288 = 93 + 84 + 111$ bits, for any clock $t\geq 0$.
The keystream bits are known by the attackers starting with clock
$T = 4\cdot 288 = 1152$. The key and the initial vector of \Tri\ are 80 bit
vectors and they form together 160 bits of an initial state. The remaining 128 bits
are fixed ones.

By Corollary \ref{dinvco}, we obtain that the system (\ref{tri}) is invertible
with inverse system
\begin{equation*}
\left\{
\begin{array}{rcl}
x(93) & = & y(0) + x(66) + y(78) + x(91) x(92), \\
y(84) & = & z(0) + y(69) + z(87) + y(82) y(83), \\
z(111) & = & x(0) + z(66) + x(69) + z(109) z(110). \\
\end{array}
\right.
\end{equation*}
This allows an algebraic attack to the $T$-state instead of the initial state
containing the key and the initial vector. The problem with such an attack
is the high number (288) of variables in the set
\[
\bX = \{x(0),\ldots,x(92),y(0),\ldots,y(83),z(0),\ldots,z(100)\}.
\]
Indeed, to solve the polynomial system obtained by means of Theorem \ref{keyeq}
using a guess-and-determine strategy, one has to evaluate a number
of variables that exceeds the length of the key which is 80 bit.
In other words, providing that all solving computations for any guess
can actually be performed in a reasonable time, one has a complexity
which is greater than the key recovery by exhaustive search. An experimental evidence
of this is contained, for instance, in \cite{EVP,HL}. We also tried ourselves with
a time limit of one hour for the \Gr\ bases computations. We plan to further
improve our guess-and-determine strategies to tackle \Tri.

Therefore, we present here all optimizations and computational data that we have
obtained for a well-studied simplified version of \Tri\ cipher which is called \Biv.
For this cipher we obtain a running time for an algebraic attack which improves
a previous one \cite{EVP} and it is much better than brute force.

The explicit difference system defining \Biv\ are the following two quadratic equations
\begin{equation}
\label{biv}
\left\{
\begin{array}{rcl}
x(93) & = & y(0) + y(15) + x(24) + y(1) y(2), \\
y(84) & = & x(0) + y(6) + x(27) + x(1) x(2), \\
\end{array}
\right.
\end{equation}
and its keystream polynomial is
\begin{equation*}
\begin{array}{rcl}
f & = & x(0) + x(27) + y(0) + y(15).
\end{array}
\end{equation*}
In this case, the $t$-states are vectors of $93 + 84 = 177$ bits and the keystream
starts at clock $T = 4\cdot 177 = 708$. Again, the key and the initial vector
are $80 + 80 = 160$ bits of the initial state. Corollary \ref{dinvco} implies that
the system (\ref{biv}) is invertible with inverse
\begin{equation*}
\left\{
\begin{array}{rcl}
x(93) & = & y(0) + x(66) + y(78) + x(91) x(92), \\
y(84) & = & x(0) + x(69) + y(69) + y(82) y(83). \\
\end{array}
\right.
\end{equation*}
Consider the polynomial algebras $R = \KK[X]$ where $X = \bigcup_{t\geq 0} \{x(t),y(t)\}$
and $\bR = \KK[\bX]$ where $\bX = \{x(0),\ldots,x(92),y(0),\ldots,y(83)\}$.
If $R$ is endowed with a clock-based monomial ordering, we have that
\begin{equation*}
\begin{gathered}
G = \{ x(93) + y(0) + y(15) + x(24) + y(1) y(2), \\
       y(84) + x(0) + y(6) + x(27) + x(1) x(2) \}
\end{gathered}
\end{equation*}
is a difference \Gr\ basis of the difference ideal $I = \langle G\rangle_\sigma\subset R$.
Consider the polynomials $f'_t = \bT^t(f) = \NF_I(\sigma^t(f))\in \bR$ ($t\geq 0$)
and let $b:\NN\to\KK$ be a keystream. Since the system (\ref{biv}) is invertible, we can
attack the $T$-state by the ideal
\[
J''_B = \sum_{0\leq t\leq B-T} \langle f'_t + b(T + t) \rangle\subset \bR.
\]
We have found experimentally that if the number $B-T+1$ of known values of the keystream is
approximately 180, one has a unique $\KK$-solution in $V_\KK(J''_B)$. In this case,
the maximal degree of the generators of $J''_B$ is 3. Even though we use a DegRevLex
monomial ordering on $\bR$, a \Gr\ basis of the ideal $J''_B + L\subset \bR$
where
\[
L = \langle x(0)^2 + x(0), \ldots, x(92)^2 + x(92), y(0)^2 + y(0), \ldots,
y(83)^2 + y(83) \rangle
\]
seems to be hard to compute. Indeed, we have experimented that the number of S-polynomials
for such a computation increases in a very fast way affecting time and space complexity.
We therefore use a guess-and-determine strategy to obtain solving times that we can
actually determine.

Since all shifts $\sigma^t(f)$ ($0\leq t\leq 65$) of the keystream polynomial
$f = x(0) + x(27) + y(0) + y(15)$ are normal modulo $I$, we have 66 linear
polynomials $\sigma^t(f) - b(T + t) = f'_t - b(T + t)\in J''_B$. Because the clocks
of the variables in $f$ are all multiples of 3, we can divide these polynomials
into 3 sets of 22 linear polynomials, namely
\[
S_i = \{\sigma^t(f) - b(T + t)\mid 0\leq t\leq 65, t\equiv i\ \hbox{mod}\ 3\}\ (0\leq i\leq 2).
\]
By performing Gaussian elimination over $S_i$, we obtain 22 pivot variables and 36 free
variables. In other words, for any set $S_i$ the evaluation of 36 variables implies
the evaluation of $36 + 22 = 58$ variables. This is a good trick that was first
observed in \cite{MB}. In our computations, we choose the set $S_2$, that is, we guess
the following 36 free variables
\[
x(68), x(71), \ldots, x(92), y(2), y(5), \ldots, y(80)
\]
and we obtain the evaluation of the 22 pivot variables
\[
x(2), x(5), \ldots, x(65).
\]
Moreover, note that one has the polynomial $f'_{68} - b(T + t)\in J''_B$ where
\[
f'_{68} = y(83) + x(68) + y(68) + x(26) + y(17) + y(4) y(3) + y(2).
\]
By guessing the variables $y(3),y(4)$ together with the previous 36 ones, one
obtains in fact the evaluation of the variable $y(83)$, that is, a total of $61$
evaluations out of the 177 variables of the algebra $\bR$. This is enough
to have \Gr\ bases computations that last only a few tenths of seconds.
Precisely, our guess-and-determine strategy for \Biv\ consists in computing the \Gr\ bases
of all ideals $J''_B + L + E_{\alpha_1,\ldots,\alpha_{38}}\subset \bR$, where
\begin{equation}
\label{guess}
\begin{gathered}
E_{\alpha_1,\ldots,\alpha_{38}} = \langle x(68) + \alpha_1, x(71) + \alpha_2,
\ldots, x(92) + \alpha_9, y(2) + \alpha_{10}, \\
y(5) + \alpha_{11}, \ldots, y(80) + \alpha_{36},
y(3) + \alpha_{37}, y(4) + \alpha_{38} \rangle
\end{gathered}
\end{equation}
and the vector $(\alpha_1,\ldots,\alpha_{38})$ ranges in the space $\KK^{38}$.

We propose now tables where we compare the solving time to obtain the set
$V_\KK(J''_B + E_{\alpha_1,\ldots,\alpha_{38}})$ by using \Gr\ bases and SAT solvers.
For \Gr\ bases, we make use of two main implementations of the Buchberger
algorithm that are available in the computer algebra system {\sc Singular} \cite{DGPS},
namely {\sc std} and {\sc slimgb}. We have decided to use this free and open-source
system because of our long experience with it, to compare with a previous attack \cite{EVP}
based on the same \Gr\ bases implementations and finally because different implementations
would affect the average running time $a\cdot 2^{37}$ only by the factor $a$ corresponding
to average solving time. The considered SAT solvers are {\sc minisat} \cite{ES} and
{\sc cryptominisat} \cite{So} which are widely used in cryptanalysis.

We have carried out the computations on a server: Intel(R) Core(TM) $i7-8700$
CPU @ $3.20$GHz, $6$ Cores, $12$ Threads, $32$ GB RAM with a Debian based Linux
operating system. In our tables, we abbreviate milliseconds and seconds by ms and s,
respectively.

\begin{table}[ht!]
\centering
\caption{$90\%$ Confidence Interval for timings with random guesses}
\label{randGuesses}
\begin{tabular}{|c|c|c|c|c|}
\hline
$\#$ ks bits & slimgb (ms) & std (ms) & MiniSat (s) & CrMiniSat (s) \\
\hline
$180$ & $(160,~195)$ & $(326,~397)$ & $(9.33,~84.03)$ & $(9.61,~40.39)$ \\
\hline
$185$ & $(159,~170)$ & $(332,~367)$ & $(6.88,~60.24)$ & $(7.33,~28.41)$ \\
\hline
$190$ & $(119,~134)$ & $(353,~411)$ & $(6.94,~63.98)$ & $(9.59,~38.26)$ \\
\hline
$195$ & $(123,~138)$ & $(342,~397)$ & $(6.55,~67.57)$ & $(7.69,~25.26)$ \\
\hline
\end{tabular}
\end{table}

\begin{table}[ht!]
\centering
\caption{$90\%$ Confidence Interval for timings with correct guess}
\label{corrGuess}
\begin{tabular}{|c|c|c|c|c|}
\hline
$\#$ ks bits & slimgb (ms) & std (ms) & MiniSat (s) & CrMiniSat (s) \\
\hline 
$180$ & $(172,~187)$ & $(330,~350)$ & $(1.21,~53.81)$ & $(2.71,~33.32)$ \\
\hline
$185$ & $(178,~191)$ & $(352,~388)$ & $(0.15,~49.59)$ & $(0.24,~15.85)$ \\
\hline
$190$ & $(127,~145)$ & $(351,~411)$ & $(0.81,~24.56 )$ & $(0.23,~31.52)$ \\
\hline
$195$ & $(122,~135)$ & $(328,~348)$ & $(1.08,~36.50 )$ & $(6.72,~19.92)$ \\
\hline
\end{tabular}
\end{table}

In both tables, the rows correspond to different choices of the number of keystream
bits that are used for the attack. The second and third columns present the $90\%$
confidence intervals for Gr\"obner bases timings that are obtained by {\sc slimgb}
and {\sc std}. The fourth and fifth columns provide the intervals for SAT solvers
timings corresponding to {\sc minisat} and {\sc cryptominisat}.

In Table \ref{randGuesses}, the confidence intervals for \Gr\ bases are computed
for $2^4$ different random (key, iv)-pairs and $2^{10}$ different random guesses
of the $38$ variables in (\ref{guess}) for each (key, iv)-pair.
In other words, the confidence intervals contain 90\% of the timings that are
obtained by a total of $2^{14}$ computations. The intervals for SAT solvers
are computed for the same set of $2^4$ (key, iv)-pairs and with a subset of $2^4$
different random guesses from the set that we have considered for \Gr\ bases.
The motivation of such reduction is larger total computing times for SAT solving.

Similarly, in Table \ref{corrGuess} the confidence intervals are computed for the
same $2^4$ (key, iv)-pairs of Table \ref{randGuesses} and the correct guess
of the $38$ variables corresponding to each (key, iv)-pair. 

For \Biv\ attack, the tables show that the procedure {\sc slimgb} is faster
than {\sc std}. This happens because ``slim'', that is, compact elements
in the resulting \Gr\ bases imply faster S-polynomial reductions which are
the most expensive component of these computations.
Moreover, we have that \Gr\ bases perform better than SAT solvers
for computing solutions of the polynomial systems involved in the \Biv\ attack.
This is especially true for the UNSAT case which is dominant in complexity.
About 190 keystream bits are the best choice for our attack and we conclude
that its average running time is $0.12 \cdot 2^{37}\,s \sim 2^{34}\,s$.

This result improves the timing $2^{39}\,s$ of an algebraic attack to \Biv\
which is presented in \cite{EVP}. This attack uses a guess-and-determine
strategy based on the exhaustive evaluation of 42 variables and \Gr\ bases
computations which are obtained by the same {\sc Singular} routines {\sc std}
and {\sc slimgb} running on a comparable CPU.

\section{Attacking \Kee}

We present now an illustrative example of an algebraic attack to a concrete
difference block cipher. A well-known small size block cipher is \Kee\ which
has important applications in remote keyless entry systems which are used,
for instance, by the automotive industry. \Kee\ is a proprietary cipher \cite{Kee}
whose cryptographic algorithm was created by Gideon Kuhn at the University
of Pretoria in the mid-1980s. Starting from the mid-1990s, the cipher was widely
used by car manufactures but it has begun to be cryptanalysed only in 2007.
In particular, we mention the papers \cite{CBB,CBW} where there are algebraic
attacks to \Kee\ on which this section is based. For another important class
of ``meet-in-the-middle attacks'', see \cite{Pra}.

The block cipher \Kee\ is defined by a reducible invertible difference system
over the base field $\KK = \GF(2)$, where the key subsystem consists of a single
homogeneous linear equation (LFSR). In fact, its state transition $\KK$-linear map
corresponds to a cyclic permutation matrix of period 64. In addition to the key
equation, the invertible system consists of an explicit difference cubic equation
involving a single key variable. Precisely, the \Kee\ system is the following one

\begin{equation}
\label{kee}
\left\{
\begin{array}{rcl}
k(64) & = & k(0), \\
x(32) & = & x(0) + x(16) + x(9) + x(1) + x(20) x(31) \\
& & +\ x(1) x(31) + x(20) x(26) + x(1) x(26) + x(9) x(20) \\
& & +\ x(1) x(9) + x(1) x(9) x(31) + x(1) x(20) x(31) \\
& & +\ x(9) x(26) x(31) + x(20) x(26) x(31) + k(0).
\end{array}
\right.
\end{equation}

The key, that is, the initial state of the key equation, consists therefore of 64 bits
and the plaintext and ciphertext are 32 bits vectors. In other words, any $t$-state
of \Kee\ consists of $64 + 32 = 96$ bits. The final clock of this difference block cipher
is defined as the clock $T = 8\cdot 64 + 16 = 528$. Theorem \ref{invth} provides that
the system (\ref{kee}) is invertible with the following inverse system

\begin{equation}
\label{invkee}
\left\{
\begin{array}{rcl}
k(64) & = & k(0), \\
x(32) & = & x(0) + x(31) + x(23) + x(16) + x(23) x(31) \\
& & +\ x(6) x(31) + x(1) x(31) + x(12) x(23) + x(6) x(12) \\
& & +\ x(1) x(12) + x(1) x(23) x(31) + x(1) x(12) x(31) \\
& & +\ x(1) x(6) x(23) + x(1) x(6) x(12) + k(0).
\end{array}
\right.
\end{equation}

To describe a multiple plaintext-ciphertext pairs attack to \Kee, consider two such
pairs $(v',v''), (w',w'')\in \KK^{32}\times \KK^{32}$, where
$v' = (\alpha'_0,\ldots,\alpha'_{31}), v'' = (\alpha''_0,\ldots,\alpha''_{31})$ and
$w' = (\beta'_0,\ldots,\beta'_{31}), w'' = (\beta''_0,\ldots,\beta''_{31})$.
Then, define the polynomial algebra $R' = \KK[X']$ where $X' =
\bigcup_{t\geq 0} \{k(t),x(t),y(t)\}$ and consider the following set
\begin{equation*}
\begin{gathered}
G' = \{
k(64) + k(0), \\
x(32) + x(0) + x(31) + x(23) + x(16) + x(23) x(31) + x(6) x(31) \\
+\ x(1) x(31) + x(12) x(23) + x(6) x(12) + x(1) x(12) + x(1) x(23) x(31) \\
+\ x(1) x(12) x(31) + x(1) x(6) x(23) + x(1) x(6) x(12) + k(0), \\
y(32) + y(0) + y(31) + y(23) + y(16) + y(23) y(31) + y(6) y(31) \\
+\ y(1) y(31) + y(12) y(23) + y(6) y(12) + y(1) y(12) + y(1) y(23) y(31) \\
+\  y(1) y(12) y(31) + y(1) y(6) y(23) + y(1) y(6) y(12) + k(0)
\}.
\end{gathered}
\end{equation*}
We define the difference ideal $I' = \langle G'\rangle_\sigma\subset R'$
and the linear ideal $J' = J'(0) + J'(T)\subset R'$ where
\begin{equation*}
\begin{gathered}
J'(0) = \langle x(0) + \alpha'_0,\ldots, x(31) + \alpha'_{31},
               y(0) + \beta'_0,\ldots, y(31) + \beta'_{31} \rangle, \\
J'(T) = \langle x(T) + \alpha''_0,\ldots, x(T + 31) + \alpha''_{31},
               y(T) + \beta''_0,\ldots, y(T + 31) + \beta''_{31} \rangle.
\end{gathered}
\end{equation*}

An algebraic attack to \Kee\ by the plaintext-ciphertext pairs $(v',v'')$ and
$(w',w'')$ consists in computing $V_\KK(I' + J')$. Note that we can indeed
assume that the variables clocks are bounded by $T + 31$, that is, we solve over
the finite set of variables $\bigcup_{0\leq t\leq T + 31} \{k(t),x(t),y(t)\}$.
Actually, the computation of $V_\KK(I' + J')$ is unfeasible for $T = 528$.
If we would assume that $T = 512 = 8\cdot 64$, we could use the trick of fixed
pairs which is described at the end of Section 5. Briefly, if
$(u(t),v(t))\in\KK^{64}\times\KK^{32}$ denotes the $t$-state of the system
(\ref{kee}), the trick consists in computing enough fixed points $v(0) = v(512)$
by the knowledge of the encryption function and to assume that some of them
are in fact fixed points $v(0) = v(64)$.
By means of a couple of such plaintext-ciphertext pairs $(v(0),v(0))$, we are
reduced to compute $V_\KK(I' + J')$ for $T = 64$.

The problem now is how to compute $v(512)$ from the ciphertext $v(528)$.
If we assume that the variables $k(0),\ldots,k(15)$ are evaluated by 
the correct corresponding key bits, we can apply the inverse cipher (\ref{invkee})
since these bits are the only ones that are involved in the computation of $v(512)$
from $v(528)$. Indeed, the authors of \cite{CBB,CBW} have studied a property based
on the disjoint cycles decomposition which is able to distinguish a generic permutation
of the set $\KK^{32}$ from the \Kee\ encryption function reduced to $T = 512$ clocks.
By means of this method, the cost of the search of the correct values of the variables
$k(0),\ldots,k(15)$ is assumed to be $2^{52}$ CPU clocks in the worst case. With an
optimized implementation on our CPU @ $3.20$GHz, this method should then take
$a = 2^{21}$ sec.

Once obtained the correct values $\alpha_1,\ldots,\alpha_{15}\in\KK$ of the variables
$k(0),\ldots,k(15)$, we have experimented that to compute each set
$V_\KK(I' + J' + E_{\alpha_0,\ldots,\alpha_{15}})$ where $T = 64$ and
\[
E_{\alpha_0,\ldots,\alpha_{15}} =
\langle k(0) + \alpha_0, \ldots, k(15) + \alpha_{15} \rangle
\]
one needs just few tens of milliseconds by using \Gr\ bases or SAT solvers.
Note that this improves solving times obtained in \cite{CBB,CBW} which are hundreds
of milliseconds on a similar CPU.

We conclude that the total running time of this algebraic attack to \Kee\
can be described by the formula
\begin{equation}
\label{keetime}
a + b\cdot c\cdot 2^{32} + d
\end{equation}
where $b$ is the average encryption time for $T = 528$ clocks, $c$ is the average
percentage of the plaintext space $\KK^{32}$ containing enough fixed points $v(0) = v(512)$
and $d$ is the average computing time to obtain the sets
$V_\KK(I' + J' + E_{\alpha_0,\ldots,\alpha_{15}})$ for each couple of such
fixed points. The authors of \cite{CBB,CBW} have experimented that there are $26\%$ of keys
such that $c = 60\%$. Among the computed fixed points $v(0) = v(512)$, they also
assume a good chance that at least one couple $v'(0), v''(0)$ of them is such that
$v'(0) = v'(64), v''(0) = v''(64)$. In our experiments we make use of $4$ distinct
random such ``weak keys''.

In the following table, we present then statistics of the values $b,d$. Similarly to
Section 6, the solving time $d$ is provided for \Gr\ basis algorithms and SAT solvers
and it appears in the columns corresponding to {\sc slimgb, std, minisat, cryptominisat}.
Recall that $d$ is the total computing time for solving the polynomial systems
corresponding to all couples that are obtained by computed fixed points $v(0) = v(512)$.

\begin{table}[ht!]
\centering
\caption{$90\%$ Confidence Interval for timings}
\label{corrGkee}
\begin{tabular}{|c|c|c|c|c|}
\hline
b (ms) & slimgb (ms) & std (ms) & MiniSat (ms) & CrMiniSat (ms) \\
\hline 
$(2.9,~3.2)\cdot 2^{-10}$ & $(93,~183)$ & $(27,~93)$ & $(12,~27)$ & $(16,~32)$ \\
\hline
\end{tabular}
\end{table}

In Table \ref{corrGkee}, we present $90\%$ confidence intervals corresponding to
4 weak key and the correct guess of the 16 variables corresponding to each key.
The total encryption time $b\cdot c\cdot 2^{32}$ when $c = 0.6$ is
about 2.5 hours by an executable C file. The confidence interval of $b$, that is,
the timing for a single encryption is obtained by the corresponding interval
of this total time.

For the polynomial systems involved in the \Kee\ attack, the SAT solvers seem to be
the best option. Nevertheless, note that in the tables of Section 6 and 7,
the time for computing ANF-to-CNF conversion (see, for instance, \cite{Ba})
is not considered. In particular, for \Kee\ attack this would imply that
\Gr\ bases ({\sc std} method) and SAT solvers have comparable timings. In any case,
for the total running time (\ref{keetime}) of this algebraic attack to \Kee, the dominant
timing is clearly $a = 582$ hours which confirms the results in \cite{CBB,CBW}.

\section{Conclusions}

We have shown in this paper that the notion of system of (ordinary) explicit
difference equations over a finite field is useful for modeling the class
of ``stream and block difference ciphers'' which many ciphers of application
interest belong to. The appropriate algebraic formalization of such systems
and corresponding ciphers requires the theory of difference algebras and ideals,
as well the methods of difference \Gr\ bases. This formalization allows the study
of general properties of the difference ciphers such as their invertibility and
periodicity. This study is essential to assess their security by means of suitably
defined algebraic attacks. We have illustrated this in practice using two well-known
difference ciphers, \Biv\ and \Kee, where \Gr\ bases and SAT solvers are also
compared. We plan to include different methods from the formal theory of Ordinary Difference
Equations for further improving the cryptanalysis of difference ciphers. We believe
therefore that the proposed modeling and the corresponding methods will be useful
for the development of new applicable ciphers.

\section{Acknowledgements}

We would like to thank Arcangelo La Bianca, not only for providing computational
resources for our experimental sessions, but also for having patiently supported
the first author erratic paths in problem solving.
We also thank the anonymous referees for the careful reading of the
manuscript. We have sincerely appreciated all valuable comments and suggestions
as they have significantly improved the readability of the paper.

\end{document}